%% file: main.tex
\documentclass[12pt]{article}
\usepackage{graphicx}
\usepackage{tikz}
\usepackage{amsmath,amsthm,thmtools,thm-restate,amssymb,mathrsfs,mathtools,bm}
\usepackage[OT1]{fontenc} 
\usepackage{etoc}
\usepackage{float}
\usepackage{titlesec}
\usepackage{multirow,array}
\usepackage{diagbox}
\usepackage{fancyhdr}

\titleformat{\section}
		{\scshape \large \bfseries\center}
         {\thesection}
        {0.5em}
        {}
        []
\titleformat{\subsection}
        {\normalfont\bfseries}
        {\thesubsection}
        {0.5em}
        {}
        []

\allowdisplaybreaks

\def\subsectiontitle{}
\def\subsubsectiontitle{}
\usepackage{comment}
\usepackage[backgroundcolor=orange!50!white]{todonotes}
\usepackage{tikz}
\usetikzlibrary{automata, positioning}

\usepackage{kpfonts}
\usepackage{inconsolata}

\usepackage{makecell}
\usepackage{cellspace}
\usepackage{enumitem}
\usepackage[margin=1in]{geometry}
\usepackage[colorlinks,breaklinks=true]{hyperref}
\usepackage{breakcites}

\usepackage{xcolor}
\AtBeginDocument{%
	\hypersetup{%
    linkcolor={red!50!black},%
    citecolor={blue!50!black},%
    urlcolor={blue!80!black}%
}%
}%

\usepackage[nameinlink,noabbrev,sort,capitalise]{cleveref}
\makeatletter
\def\ps@pprintTitle{%
 \let\@oddhead\@empty
 \let\@evenhead\@empty
 \def\@oddfoot{\emph{Very preliminary version}\hfill\emph{This draft: \today}}%
 \let\@evenfoot\@oddfoot}
\makeatother
\usepackage{etoolbox}
\patchcmd{\pprintMaketitle}
  {\fi\hrule}
  {\fi\ifvoid\extrainfobox\else\unvbox\extrainfobox\par\vskip10pt\fi\hrule}
  {}{}

\newsavebox\extrainfobox

\crefname{proposition}{Proposition}{Propositions}

\crefname{theorem}{Theorem}{Theorems}

\crefname{corollary}{Corollary}{Corollaries}

\newtheorem{lemma}{Lemma}
\crefname{lemma}{Lemma}{Lemmas}

\crefname{assumption}{Assumption}{Assumptions}

\crefname{axiom}{Axiom}{Axioms}

\newtheorem{definition}{Definition}
\crefname{definition}{Definition}{Definitions}

\theoremstyle{remark}
\newtheorem*{remark}{Remark}

\theoremstyle{definition}

\crefname{example}{Example}{Examples}

\crefname{problem}{Problem}{Problems}

\theoremstyle{claim}

\crefname{claim}{Claim}{Claimss}

\usepackage{indentfirst}

\allowdisplaybreaks

\let\oldfootnote\footnote
\renewcommand\footnote[1]{\oldfootnote{\hspace{.4mm}#1}}

\makeatletter
\renewenvironment{proof}[1][\proofname] {\par\pushQED{\qed}\normalfont\topsep6\p@\@plus6\p@\relax\trivlist\item[\hskip\labelsep\bfseries#1\@addpunct{.}]\ignorespaces}{\popQED\endtrivlist\@endpefalse}
\makeatother

\let\oldFootnote\footnote
\newcommand\nextToken\relax

\renewcommand\footnote[1]{%
    \oldFootnote{#1}\futurelet\nextToken\isFootnote}

\newcommand\isFootnote{%
    \ifx\footnote\nextToken\textsuperscript{,}\fi}

\usepackage[american]{babel}
\usepackage[authoryear,round]{natbib}
\usepackage{todonotes}

\usepackage{amstext} 
\usepackage{array}   
\newcolumntype{L}{>{$}l<{$}} 

\newcommand{\pf}[1]{\left[\begin{smallmatrix}#1\end{smallmatrix}\right]}
\newcommand{\vpf}[1]{\left(\begin{smallmatrix}#1\end{smallmatrix}\right)}

\begin{document}

\title{The Impossibility of a Gerrymander-Proof Representative Democracy}

\date{This version: October 16, 2026}

\author{\makebox[.25\linewidth]{{John Mori}\thanks{University of Chicago Department of Economics; email: \protect\texttt{johnmori@uchicago.edu}. I am grateful to Ben Brooks, Phil Reny, and Joe Root for guidance and feedback. I also thank Chris Chambers and the audience at the theory workshop at the University of Chicago for helpful discussion.}}}

\maketitle

\input{draft/abstract}
\input{draft/intro}
\input{draft/literature}
\input{draft/certain}
\input{draft/uncertain}

\input{draft/independence}

\small 
\setlength{\bibsep}{0pt}
\bibliography{draft/bibliography}

\appendix

\input{draft/proof_certain}
\input{draft/proof_rank}
\input{draft/proof_uncertain}

\end{document}

%% file: draft/abstract.tex
\begin{abstract}
    A representative democracy is immune to gerrymandering if it satisfies \citeauthor{chambers2008consistent}' (\citeyear{chambers2008consistent}) \emph{representative consistency}. We examine preference aggregation and show that representative consistency is incompatible with three other normative desiderata --- weak Pareto, anonymity, and neutrality. We show this impossibility result both in a setting with ordinal preferences and in a setting with expected utility preferences.
\end{abstract}

%% file: draft/intro.tex
\section{Introduction}

A representative democracy is susceptible to \emph{gerrymandering} if political parties can manipulate district boundaries to influence the outcomes of elections. Some voting systems are more susceptible to gerrymandering than others. For instance, single-member, winner-take-all systems are highly susceptible to gerrymandering; in a two-party system, gerrymandering can nearly double the representation of a particular party. On the other hand, multiple-member, proportional representation systems alleviate the influence of gerrymandering. Are there systems that are \emph{gerrymander-proof}, i.e. immune to gerrymandering?

We show that in a preference aggregation framework, the answer is negative --- there does not exist a social welfare function that is immune to gerrymandering and satisfies other normative principles of preference aggregation: efficiency, anonymity, and neutrality.

\cite{chambers2008consistent} proposes a notion of gerrymander-proofness called \emph{representative consistency}. In our preference aggregation framework, a representatively consistent social welfare function is invariant to the partitioning of individuals in a two-stage aggregation process: individuals are partitioned into blocks, akin to how voters are divided into districts. First, the social welfare function aggregates preferences in each block, as voters in each district elect representatives for their district. Second, the social welfare function aggregates the population-weighted aggregated preferences into social preferences, as representatives determine social policy in a legislature or elect a head of government. Representative consistency requires that the resulting social preference of this two-stage process is independent of the initial partitioning of individuals. In other words, gerrymandering district boundaries does not change the eventual policies of a government.

The other axioms we consider are pillars of preference aggregation. An efficient social welfare function is responsive to agreements in individual preferences --- if every individual prefers an alternative to another, society should as well. An anonymous social welfare function cannot weigh the preferences of some individual more heavily than those of another individual. A neutral social welfare function cannot favor an alternative over another in the aggregation procedure itself.

We present the impossibility theorem both in an Arrovian setting (Section \ref{sec.ord}), where individuals have ordinal preferences over a set of alternatives, and in a Harsanyi setting (Section \ref{sec.eu}), where individuals have expected utility preferences over lotteries over alternatives. We then demonstrate the independence of our axioms in Section \ref{sec.ind}. All proofs are in the appendix.

%% file: draft/literature.tex
\textbf{Related Literature.} \cite{chambers2008consistent} introduces the axiom representative consistency (while noting precedents in the literature discussed below). Chambers examines a setting where each individual votes for one alternative among a finite set of alternatives, and a \emph{rule} maps a profile of votes to a single vote. He shows that any rule that satisfies unanimity, anonymity, and representative consistency must be a \emph{partial priority rule}, which corresponds to a meet-semilattice on the set of alternatives and, for any profile of votes, takes the meet of the votes in the profile. \cite{chambers2009axiomatic} deems this prior result an impossibility theorem, and provides a resolution by considering infinite sets of alternatives. We show that the impossibility returns in the preference aggregation setting when efficiency and neutrality axioms are added, even when the set of alternatives is infinite.

Representative consistency is closely related to the decomposability axiom in the characterization of the quasi-arithmetic mean by \cite{kolmogorov1930notion} and \cite{Nagumo1930}.\footnote{This axiom is also called ``associativity of means'' and ``barycentric associativity.''} One can view the exercise of preference aggregation as taking the ``mean'' of preferences, and naturally examine existing properties of means. We demonstrate that representative consistency, the analogue of decomposability, is inconsistent with fundamental axioms in preference aggregation.

\cite{blackorby1984social} introduce a condition they call ``the population substitution principle,'' which is the analogue of the decomposability axiom in their setting of ordering vectors of utility. There, utilities are interpersonally comparable, and the population substitution principle is used to characterize several utilitarian orders. Utilities are not interpersonally comparable in our setting of preference aggregation.

There are other impossibility theorems for gerrymandering. \cite{puppe2015axiomatic} and \cite{alexeev2018impossibility} both examine voting in a two-party system and include geographical constraints, unlike our abstract setting of preference aggregation.

%% file: draft/certain.tex
\section{Ordinal Preferences} \label{sec.ord}
Let $I$ be a set of individuals where $|I| \geq 3$, and it can be either finite or infinite. Let $\mathcal{I}$ be the set of nonempty, finite subsets of $I$. 

Let $X$ be a set of outcomes, and it can be either finite or infinite. Each individual $i \in I$ has preferences $\succsim_i \subseteq X \times X$ that are a weak order, i.e., complete and transitive. Let $\mathcal{R}$ be the set of preferences. Let ${\sim^*}$ denote complete indifference. We will often write preferences visually. For $X = \{w,x,y,z\}$, write $\pf{y\\x,z\\w}$ for the preference $\succsim \in \mathcal{R}$ for which $y \succ x \sim z \succ w$.

To study representative consistency, our social welfare function takes as input not only preference profiles of the entire population, but also preference profiles of any subset of the population. For a subpopulation $J \in \mathcal{I}$, let $\succsim_J$ denote a generic preference profile for $J$, and let $\mathcal{R}^J$ denote the set of preference profiles for $J$. Thus a generic social welfare function is $f: \bigcup_{J \in \mathcal{I}} \mathcal{R}^J \rightarrow \mathcal{R}$. Let $\succsim_{f(\succsim_J)}$, $\succ_{f(\succsim_J)}$, and $\sim_{f(\succsim_J)}$ denote the weak, strict, and indifference parts of $f(\succsim_J)$.

For preferences $\succsim \in \mathcal{R}$ and subpopulation $J \in \mathcal{I}$, let $\succsim^J$ denote the preference profile where each individual $i \in J$ has preferences $\succsim$. We can now state our main axiom of study, introduced by \cite{chambers2008consistent}:
\begin{definition}
    Representative consistency. For any $J \in \mathcal{I}$, any $\succsim_J \in \mathcal{R}^J$, and any partition $\{K_1, \dots, K_m\}$ of $J$, $f(\succsim_J) = f\left(f(\succsim_{K_1})^{K_1}, \dots, f(\succsim_{K_m})^{K_m}\right)$.
\end{definition}

The next axiom is the standard weak Pareto.

\begin{definition}
    Weak Pareto. For any $J \in \mathcal{I}$, any $\succsim_J \in \mathcal{R}^J$, and any $x,y \in X$, if $x \succ_i y$ for all $i \in J$, then $x \succ_{f(\succsim_J)} y$.
\end{definition}

Since the social welfare function can take as input preference profiles of subpopulations, anonymity requires that the social welfare function is independent of the specific subpopulation whose preferences are being aggregated:
\begin{definition}
    Anonymity. For any $J, K \in \mathcal{I}$, any bijection $\sigma: J \rightarrow K$, and any $\succsim_J \in \mathcal{R}^J$, $\succsim_{K}' \in \mathcal{R}^K$ for which $\succsim_i = \succsim_{\sigma(i)}'$ for all $i \in J$, $f(\succsim_J) = f(\succsim_K')$.
\end{definition}

Neutrality requires that the social welfare function does not favor any alternative in the aggregation procedure itself. For any permutation $\rho: X \rightarrow X$ and preferences $\succsim \in \mathcal{R}$, let $\rho(\succsim)$ be the preference such that for any $x, y \in X$, $\rho(x) \rho(\succsim)\rho(y) \iff x \succsim y$. For any subpopulation $J \in \mathcal{I}$ and preference profile $\succsim_J \in \mathcal{R}^{J}$, let $\rho(\succsim_J) = (\rho(\succsim_{j_1}), \dots, \rho(\succsim_{j_n}))$ for $J = \{j_1, \dots, j_n\}$. 

\begin{definition}
    Neutrality. For any $J \in \mathcal{I}$, any $\succsim_J \in \mathcal{R}^{J}$, and any permutation $\rho: X \rightarrow X$, $f(\rho(\succsim_J)) = \rho(f(\succsim_J))$.
\end{definition}

\begin{restatable}{theorem}{thmord} \label{thm.ord}
    Suppose $I$ is countably infinite and $|X| \geq 4$. There is no social welfare function that satisfies representative consistency, weak Pareto, anonymity, and neutrality.
\end{restatable}

We demonstrate the independence of our axioms in Section \ref{sec.ind}. 

We assume that $I$ is countably infinite since our proof makes use of Theorem 1 of \cite{chambers2008consistent}, which requires this assumption.

The bound on the number of outcomes is tight. If $|X| \leq 3$, these axioms characterize a pair of rules which we call \emph{rank priority rules}. The rank priority rules have an intuitive definition for any finite set of alternatives $X$ --- not just for $|X| \leq 3$ --- so we will define them for finite $X$. Our axioms only characterize the rank priority rules when $|X| \leq 3$. When $|X| \geq 4$, the rank priority rules fail representative consistency.

We first define the \emph{top} rank priority rule. The top rank priority rule builds the social preference by determining the indifference classes from the top sequentially. For any $\succsim \in \mathcal{R}$, let $I_k(\succsim)$ denote the $k$th indifference class from the top. For any $\succsim \in \mathcal{R}$ and any subset of alternatives $Y \subseteq X$, write $\succsim \mid_{Y}$ for the induced weak order on $Y$. Let the \emph{top rank priority rule} $f_{trp}$ be defined as follows: for any $J \in \mathcal{I}$ and $\succsim_J \in \mathcal{R}^{J}$, $I_1(f_{trp}(\succsim_J)) = \bigcup_{i \in J} I_1(\succsim_i)$, $I_2(f_{trp}(\succsim_J)) = \bigcup_{i \in J} I_1(\succsim_i\mid_{X - I_1(f_{trp}(\succsim_J))})$, $I_3(f_{trp}(\succsim_J)) = \bigcup_{i \in J} I_1(\succsim_i\mid_{X - I_1(f_{trp}(\succsim_J)) - I_2(f_{trp}(\succsim_J))})$, and generally for $k \in \mathbb{N}$, $I_k(f_{trp}(\succsim_J)) = \bigcup_{i \in J} I_1(\succsim_i\mid_{X - I_1(f_{trp}(\succsim_J)) - I_2(f_{trp}(\succsim_J)) - \dots - I_{k-1}(f_{trp}(\succsim_J))})$. The \emph{bottom rank priority rule} $f_{brp}$ is defined similarly, where $I_k(\succsim)$ denotes the $k$th indifference class from the bottom. 

When $|X| = 2$, the top and bottom rank priority rules coincide: for any $J \in \mathcal{I}$ and $\succsim_J \in \mathcal{R}^J$, if every individual in $J$ has the same preference $\succsim$, then $f_{trp}(\succsim_J) = \succsim$. Otherwise, $f_{trp}(\succsim_J) = {\sim^*}$.

\begin{restatable}{proposition}{proprank} \label{prop.rank}
    Suppose $I$ is countably infinite and $|X| \leq 3$. A social welfare function $f$ satisfies representative consistency, weak Pareto, anonymity, and neutrality if and only if it is a rank priority rule.
\end{restatable}

For $|X| \geq 4$, the rank priority rules do not satisfy representative consistency. We illustrate this for the top rank priority rule when $X = \{w,x,y,z\}$ and $J = \{1,2,3\}$. Note that
\begin{align*}
    f_{trp}\left(\pf{w\\x\\z\\y}, \pf{w\\y\\x\\z}\right) &= \pf{w\\x,y\\z}, \\
    f_{trp}\left(\pf{w\\x\\z\\y}, \pf{w\\y\\x\\z}, \pf{x\\w\\y\\z}\right) &= \pf{w,x\\y,z}, \text{ and} \\
    f_{trp}\left(\pf{w\\x,y\\z}, \pf{w\\x,y\\z}, \pf{x\\w\\y\\z}\right) &= \pf{w,x\\y\\z}.
\end{align*}

But
\begin{align*}
    f_{trp}\left(\pf{w\\x\\z\\y}, \pf{w\\y\\x\\z}, \pf{x\\w\\y\\z}\right) & \neq f_{trp}\left(\pf{w\\x,y\\z}, \pf{w\\x,y\\z}, \pf{x\\w\\y\\z}\right) \\
    & = f_{trp} \left(f_{trp}\left(\pf{w\\x\\z\\y}, \pf{w\\y\\x\\z}\right),
        f_{trp}\left(\pf{w\\x\\z\\y}, \pf{w\\y\\x\\z}\right),
        f_{trp}\left(\pf{x\\w\\y\\z} \right)\right).
\end{align*}

And so $f_{trp}$ doesn't satisfy representative consistency.

%% file: draft/uncertain.tex
\section{Expected Utility Preferences} \label{sec.eu}
At this point one might think that the impossibility is due to the fact that ordinal preferences over certain outcomes are not ``fine-grained'' enough to accommodate representative consistency. After all, the quasi-arithmetic means --- which satisfy decomposability --- are characterized for real numbers. Suppose we instead consider ``cardinal'' preferences that can be represented by real-valued utilities unique up to positive affine transformations, namely expected utility preferences.

In this section we show that the impossibility remains even when we consider expected utility preferences.

Denote by $\Delta(X)$ the set of probability distributions with finite support, or lotteries, over $X$. We denote the degenerate lottery on $x$ by $x$. Each individual $i \in I$ has preferences $\succsim_i$ over $\Delta(X)$. We assume that each individual's preferences admit a von Neumann--Morgenstern (vNM) expected utility representation --- say that $\succsim_i$ is \emph{vNM} if there exists a utility index $u_i: X \rightarrow \mathbb{R}$ such that the \emph{EU representation} $P \mapsto \sum_{x \in \mathrm{supp}(P)}u_i(x)P(x)$ represents $\succsim_i$. Let $\mathcal{V}$ denote the set of vNM preferences on $\Delta(X)$. For a utility index $u_i$ and a lottery $P$, write $u_i(P) := \sum_{x \in \mathrm{supp}(P)} u_i(x) P(x)$. Again, let $\sim^*$ denote complete indifference.

Then, a generic social welfare function for vNM preferences is $\phi: \bigcup_{J \in \mathcal{I}} \mathcal{V}^J \rightarrow \mathcal{V}$.

Let's state the axioms in the vNM setting:

\begin{definition}
    Representative consistency. For any $J \in \mathcal{I}$, any $\succsim_J \in \mathcal{V}^J$, and any partition $\{K_1, \dots, K_m\}$ of $J$, $\phi(\succsim_J) = \phi\left(\phi(\succsim_{K_1})^{K_1}, \dots, \phi(\succsim_{K_m})^{K_m}\right)$.
\end{definition}

\begin{definition}
    Weak Pareto. For any $J \in \mathcal{I}$, any $\succsim_J \in \mathcal{V}^J$, and any $P, Q \in \Delta(X)$, if $P \succ_i Q$ for all $i \in J$, then $P \succ_{\phi(\succsim_J)} Q$.
\end{definition}

\begin{definition}
    Anonymity. For any $J, K\in \mathcal{I}$, any bijection $\sigma: J \rightarrow K$, and any $\succsim_J \in \mathcal{V}^J$, $\succsim_{K}' \in \mathcal{V}^{K}$ for which $\succsim_i = \succsim_{\sigma(i)}'$ for all $i \in J$, $\phi(\succsim_J) = \phi(\succsim_K')$.
\end{definition}

For a permutation $\rho: X \rightarrow X$ and lottery $P \in \Delta(X)$, denote by $\rho(P)$ the lottery which assigns $\rho(P)(\rho(x)) = P(x)$. For any preferences $\succsim \in \mathcal{V}$, let $\rho(\succsim)$ be the preference such that for any $P, Q \in \Delta(X)$, $\rho(P) \rho(\succsim)\rho(Q) \iff P \succsim Q$. For any subpopulation $J \in \mathcal{I}$ and preference profile $\succsim_J \in \mathcal{V}^{J}$, let $\rho(\succsim_J) = (\rho(\succsim_{j_1}), \dots, \rho(\succsim_{j_n}))$ for $J = \{j_1, \dots, j_n\}$. 

\begin{definition}
    Neutrality. For any $J \in \mathcal{I}$, any $\succsim_J \in \mathcal{V}^{J}$, and any permutation $\rho: X \rightarrow X$, $\phi(\rho(\succsim_J)) = \rho(\phi(\succsim_J))$.
\end{definition}

\begin{restatable}{theorem}{thmeu} \label{thm.eu}
    Suppose $|X| \geq 3$. There is no social welfare function that satisfies representative consistency, weak Pareto, anonymity, and neutrality.
\end{restatable}

The bound on the number of outcomes is tight. When $|X| = 2$, the set of expected utility preferences coincides with the set of ordinal preferences.\footnote{If $X = \{x,y\}$, there are three distinct expected utility preferences: $x \succ y$, $x \prec y$, and $x \sim y$. These are also the three distinct ordinal preferences.} And so the rank priority rule introduced in Section \ref{sec.ord} satisfies the axioms. 

In this setting, if we strengthen weak Pareto to strong Pareto, we can drop anonymity and neutrality from the impossibility theorem --- representative consistency and strong Pareto are directly in conflict.

\begin{definition}
    Strong Pareto. For any $J \in \mathcal{I}$, any $\succsim_J \in \mathcal{V}^J$, and any $P, Q \in \Delta(X)$, i) if $P \succsim_i Q$ for all $i \in J$, then $P \succsim_{\phi(\succsim_J)} Q$, and ii) if furthermore there exists $j \in J$ for whom $P \succ_j Q$, then $P \succ_{\phi(\succsim_J)} Q$.
\end{definition}

\begin{restatable}{proposition}{propeustrong} \label{prop.eustrong}
    Suppose $|X| \geq 3$. There is no social welfare function that satisfies representative consistency and strong Pareto.
\end{restatable}

The bound on the number of outcomes in Proposition \ref{prop.eustrong} is also tight. Suppose $X = \{x,y\}$, and write $p_x$ and $p_y$ for the preferences with $x \succ y$ and $y \succ x$, respectively. The \emph{$x$-priority rule} $\phi_x$ is defined as follows: for any $J \in \mathcal{I}$ and $\succsim_J \in \mathcal{V}^J$,
\begin{align*}
    \phi_x(\succsim_J) = \begin{cases}
        p_x & \text{if } {\succsim_i} = p_x \text{ for some } i \in J, \\
        p_y & \text{if } {\succsim_i} \neq p_x \text{ for all } i \in J \text{ and } {\succsim_j} = p_y \text{ for some } j \in J, \\
        {\sim^*} & \text{if } {\succsim_i} = {\sim^*} \text{ for all } i \in J.
    \end{cases}
\end{align*}
The \emph{$y$-priority rule} $\phi_y$ is defined similarly, with the roles of $x$ and $y$ exchanged. Since expected utility preferences and ordinal preferences coincide when $|X| = 2$, the following result also holds in the ordinal setting.

\begin{restatable}{proposition}{propxpriority} \label{prop.xpriority}
    Suppose $X = \{x,y\}$. A social welfare function satisfies representative consistency and strong Pareto if and only if it is the $x$-priority rule or the $y$-priority rule.
\end{restatable}

One interpretation for this rule is preserving the status quo. Suppose $x$ is the status quo. If anyone strictly prefers the status quo to an alternative $y$, then the status quo remains.

%% file: draft/independence.tex
\section{Independence of Axioms} \label{sec.ind}

We demonstrate the independence of our axioms for Theorems \ref{thm.ord} and \ref{thm.eu} under the assumption that $X$ is finite.

\textbf{Representative consistency}. In the ordinal setting, consider the Borda count social welfare function $f_{bc}$.

Consider any $J \in \mathcal{I}$ and $\succsim_J \in \mathcal{R}^J$. For each $i \in J$, let $b_i \in \mathbb{N}^{|X|}$ be the vector whose $x$-th component is $b_{i}^x = |\{y \in X: x \succsim_i y\}|$.\footnote{Alternative tie-breaking methods can be used to define Borda count. These other versions of Borda count also don't satisfy representative consistency.} Let $b_0 = \sum_{i \in J}b_i$, and let
$f_{bc}(\succsim_J)$ be defined by $x \succsim_{f_{bc}(\succsim_J)} y \iff b_0^x \geq b_0^y$.

Borda count doesn't satisfy representative consistency: pick any distinct alternatives $x, y\in X$ and let $Y = X\backslash \{x,y\}$. We have
\begin{align*}
    f_{bc}\left(\pf{x\\y,Y}, {\sim^*}, \pf{y\\x,Y}\right) &= \pf{x,y\\Y}, \\
    f_{bc}\left(\pf{x\\y,Y}, \pf{x\\y,Y}, \pf{y\\x,Y}\right) &= \pf{x\\y\\Y}, \text{ and} \\
    f_{bc}\left(\pf{x\\y,Y}, {\sim^*}\right) &= \pf{x\\y,Y}.
\end{align*}
    Thus,
\begin{align*}
    f_{bc}\left(\pf{x\\y,Y}, {\sim^*}, \pf{y\\x,Y}\right) = \pf{x,y\\Y} \neq \pf{x\\y\\Y} = f_{bc}\left(f_{bc}\left(\pf{x\\y,Y}, {\sim^*}\right), f_{bc}\left(\pf{x\\y,Y}, {\sim^*}\right), f_{bc}\left(\pf{y\\x,Y}\right)\right), 
\end{align*}
and representative consistency is not satisfied.

For any alternatives $x, y \in X$ such that $x \succ_i y$ for all $i \in J$, $b_{0}^x > b_{0}^y$, and so Borda count satisfies weak Pareto. The definition of Borda count doesn't treat individuals or alternatives differently, so it satisfies anonymity and neutrality.

In the expected utility setting, consider relative utilitarianism $\phi_{ru}$.\footnote{See \cite{karni1998impartiality}, \cite{dhillon1999relative}, \cite{segal2000let}, and \cite{borgers2017revealed}.} For any $J \in \mathcal{I}$ and $\succsim_J \in \mathcal{V}^J$, $\phi_{ru}(\succsim_J)$ is represented by the utility index $u_0 = \sum_{i \in J} \bar{u}_{i}$, where each $\bar{u}_{i}$ is a utility index for $\succsim_i$ and is normalized such that $\min_{x \in X} \bar{u}_{i}(x) = 0$ and, if $\succsim_i$ is not complete indifference, $\max_{x \in X} \bar{u}_{i}(x) = 1$.

Relative utilitarianism doesn't satisfy representative consistency: pick any distinct $x, y \in X$. Let $Y = X \backslash \{x, y\}$. Write $\vpf{x\\Y\\y} \in \mathcal{V}$ for the vNM preference that is represented by the utility index
\begin{align*}
    u(w) = \begin{cases}
        1 & \text{if $w=x$} \\
        0 & \text{if $w \in Y$} \\
        -1 & \text{if $w=y$,}
    \end{cases}
\end{align*}
and similarly for $\vpf{y\\Y\\x}$.

   Note that
\begin{align*}
    \phi_{ru}\left(\vpf{x\\Y\\y}, {\sim^*}, \vpf{y\\Y\\x}\right) &= {\sim^*}, \\
    \phi_{ru}\left(\vpf{x\\Y\\y}, \vpf{x\\Y\\y}, \vpf{y\\Y\\x}\right) &= \vpf{x\\Y\\y}, \text{ and} \\
    \phi_{ru}\left(\vpf{x\\Y\\y}, {\sim^*}\right) &= \vpf{x\\Y\\y}.
\end{align*}

    Thus,
\begin{align*}
    \phi_{ru}\left(\vpf{x\\Y\\y}, {\sim^*}, \vpf{y\\Y\\x}\right) = {\sim^*} &\neq \vpf{x\\Y\\y} \\
    &= \phi_{ru}\left(\phi_{ru}\left(\vpf{x\\Y\\y}, {\sim^*}\right), \phi_{ru}\left(\vpf{x\\Y\\y}, {\sim^*}\right), \phi_{ru}\left(\vpf{y\\Y\\x}\right)\right),
\end{align*}
and representative consistency is not satisfied.

In relative utilitarianism, the social utility index is a positive linear combination of individual utility indices, and so relative utilitarianism satisfies weak Pareto. Like the Borda count, the definition of relative utilitarianism doesn't treat individuals or alternatives differently, so it satisfies anonymity and neutrality.

\textbf{Weak Pareto}. A social welfare function that always returns complete indifference satisfies representative consistency, anonymity, and neutrality.

For anonymity and neutrality, we describe the rules in the ordinal setting. The expected utility setting has analogous rules.

\textbf{Anonymity}. Let $\succ^I$ be any strict total order (irreflexive, transitive, complete) on $I$. The \emph{hierarchical dictatorship} $f_{\succ^I}$ is the following rule: for any $J \in \mathcal{I}$ and $\succsim_J \in \mathcal{R}^J$, $f_{\succ^I}(\succsim_J) = \succsim_{\max_{\succ^I}(J)}$. The hierarchical dictatorship satisfies representative consistency, since for any $J \in \mathcal{I}$ and $\succsim_J \in \mathcal{R}^J$, the preference of the hierarchical dictator $\succsim_{\max_{\succ^I}(J)}$ is preserved when aggregating blocks of the partition of individuals. The hierarchical dictatorship also satisfies weak Pareto, since if the antecedent of weak Pareto holds, then it holds for the hierarchical dictator. The hierarchical dictatorship also satisfies neutrality.

\textbf{Neutrality}. Let $\succ^\mathcal{R}$ be any strict total order on $\mathcal{R}$. The \emph{total priority rule} $f_{\succ^\mathcal{R}}$ is the following rule: for any $J \in \mathcal{I}$ and $\succsim_J \in \mathcal{R}^J$, $f_{\succ^\mathcal{R}}(\succsim_J) = \max_{\succ^{\mathcal{R}}} \{\succsim_i : i \in J\}$. The total priority rule satisfies representative consistency and weak Pareto for the same reasons as the hierarchical dictatorship. It also satisfies anonymity.

%% file: draft/proof_certain.tex
\section{Proofs of Results in Section \ref{sec.ord}}

For $\succsim \in \mathcal{R}$, write $\succsim^{-1}$ for the reversed preference, i.e., for any $x, y \in X$, $x \succsim ^{-1}y \iff x \precsim y$. Similarly, for $J \in \mathcal{I}$ and $\succsim_J \in \mathcal{R}^J$, write $\succsim_J^{-1}$ for the preference profile in which each $i$'s preference is reversed.

A \emph{transposition} is a permutation that maps just two alternatives to each other (and is the identity on the rest of the alternatives). Write $(x\,y)$ for the transposition swapping $x, y \in X$. Write $(x\,y\,z)$ for the permutation that maps $x \mapsto y$, $y \mapsto z$, and $z \mapsto x$. Write $(w\,x)(y\,z)$ for the permutation swapping $w$ and $x$ with each other, and $y$ and $z$ with each other.

\subsection{Remarks and Lemmas}

In the proof we make use of another Pareto condition:
\begin{definition}
    Pareto indifference. For any $J \in \mathcal{I}$, any $\succsim_J \in \mathcal{R}^J$, and any $x, y \in X$, if $x \sim_i y$ for all $i \in J$, then $x \sim_{f(\succsim_J)} y$.
\end{definition}

\begin{remark}
    If $f$ satisfies neutrality, then $f$ satisfies Pareto indifference.
\end{remark}

To see this, consider any $J \in \mathcal{I}$, $\succsim_J \in \mathcal{R}^J$, and $x, y \in X$ such that $x \sim_i y$ for all $i \in J$. If $x \succ_{f(\succsim_J)} y$, then by neutrality the transposition $\rho = (x\,y)$ implies that $x \prec_{f(\rho(\succsim_J))} y$. Since $\rho(\succsim_J) = \succsim_J$, this is a contradiction.

We introduce another property of social welfare functions. Unanimity says that if everyone has the same preferences, then society also has those preferences.
\begin{definition}
    Unanimity. For any $J \in \mathcal{I}$ and any $\succsim \in \mathcal{R}$, $f(\succsim^J) = \succsim$.
\end{definition}

\begin{remark}
    If $f$ satisfies weak Pareto and Pareto indifference, then $f$ satisfies unanimity.
\end{remark}
The proof is a direct application of weak Pareto and Pareto indifference.

Our results in the ordinal setting make use of Theorem 1 from \cite{chambers2008consistent}. First, we review some terms. A \emph{partial order} is a binary relation that is reflexive, transitive, and antisymmetric. For a partial order $\trianglerighteq$, write $\triangleright$ for its strict part. For any set $Y$, a pair $(Y, \trianglerighteq)$ is a \emph{partially ordered set} if $\trianglerighteq$ is a partial order on $Y$. For any partially ordered set $(Y, \trianglerighteq)$ and $p, q \in Y$, $p \wedge q \in Y$ is the \emph{meet of $p$ and $q$} if it is the unique greatest lower bound for $p$ and $q$ according to $\trianglerighteq$. A \emph{meet-semilattice} is a partially ordered set for which every pair of elements has a meet. Note that for a meet-semilattice $(Y, \trianglerighteq)$, the meet $\wedge$ is associative, commutative, and idempotent.\footnote{$\wedge$ is associative if for any $p, q, r \in Y$, $(p \wedge q) \wedge r = p \wedge (q \wedge r)$. $\wedge$ is commutative if for any $p, q \in Y$, $p \wedge q = q \wedge p$. $\wedge$ is idempotent if for any $p \in Y$, $p \wedge p = p$.}

\cite{chambers2008consistent} introduces the following class of rules for a finite set $Y$: $g$ is a \emph{partial priority rule} if there exists a partial order $\trianglerighteq$ over $Y$ such that $(Y, \trianglerighteq)$ is a meet-semilattice, and for any nonempty finite $J \subseteq \mathbb{N}$ and any $\{x_i\}_{i \in J}$ where each $x_i \in Y$, $g(\{x_i\}_{i \in J}) = \bigwedge_{i \in J} x_i$.

\begin{lemma} \label{lemma.chambers}
    \citep[Theorem~1]{chambers2008consistent} Suppose $I$ is countably infinite and $X$ is finite. A social welfare function $f$ satisfies unanimity, anonymity, and representative consistency if and only if it is a partial priority rule with meet-semilattice $(\mathcal{R}, \trianglerighteq)$. 
\end{lemma}

The next lemma is a consequence of neutrality for partial priority rules. For $m \in \mathbb{N}$, a \emph{composition} of $m$ is an ordered tuple of strictly positive integers that sum to $m$. Let $C(m)$ be the set of compositions of $m$. For example, $C(3) = \{(1,1,1), (1,2), (2,1), (3)\}$. For $c \in C(m)$, let $\mathcal{R}_c$ be the set of preferences whose indifference classes --- ordered from top to bottom --- have sizes that correspond with composition $c$. For instance, for $X = \{x,y,z\}$, $\mathcal{R}_{(1,2)} = \{\pf{x\\y,z}, \pf{y\\x,z}, \pf{z\\x,y}\}$.

\begin{lemma} \label{lemma.comp}
    Suppose $X$ is finite of size $m$, $f$ is a partial priority rule with meet-semilattice $(\mathcal{R}, \trianglerighteq)$, and $f$ satisfies neutrality. For any $c \in C(m)$ and any distinct $\succsim, \succsim' \in \mathcal{R}_c$, $\succsim$ and $\succsim'$ are incomparable according to $\trianglerighteq$, i.e., it is not the case that $\succsim \trianglerighteq \succsim'$.
\end{lemma}

\begin{proof}
    Suppose towards a contradiction that $\succsim \trianglerighteq \succsim'$.
    
    Since $\succsim, \succsim' \in \mathcal{R}_c$, there exists a (non-identity) permutation $\rho: X \rightarrow X$ such that for any $x, y \in X$, $x \succsim y \iff \rho(x) \succsim' \rho(y)$. By neutrality, $\rho(\succsim) \trianglerighteq \rho(\rho(\succsim))$ and generally, for any $k \in \mathbb{N}$, $\rho^k(\succsim) \trianglerighteq \rho^{k+1}(\succsim)$. By transitivity, $\rho^k(\succsim) \trianglerighteq \rho^{\ell}(\succsim)$ for any $\ell \geq k$.

    Since $X$ is finite, there exists a $k\in \mathbb{N}$ such that $\succsim = \rho^k(\succsim)$. To see this, for each $x \in X$, there exists a $k_x \in \mathbb{N}$ such that $x = \rho^{k_x}(x)$, since $X$ is finite and $\rho$ is a bijection. Let $k$ be the least common multiple of $\{k_x\}_{x \in X}$. Then $x = \rho^k(x)$ for all $x \in X$, and so $\succsim = \rho^k(\succsim)$.

    Then $\succsim' = \rho(\succsim) \trianglerighteq \rho^k(\succsim) = \succsim$. Together with $\succsim \trianglerighteq \succsim'$, antisymmetry gives $\succsim = \succsim'$, contradicting that $\succsim$ and $\succsim'$ are distinct.
    
\end{proof}

\subsection{Proof of Theorem \ref{thm.ord}}

We restate Theorem \ref{thm.ord}.

\thmord*

\begin{proof}
    
\emph{Step 1: reduction to four alternatives.} Fix distinct $w, x, y, z \in X$ and let $Y =  \{w,x,y,z\}$. Let $\mathcal{R}\mid_Y:= \{\succsim \in \mathcal{R}: a \sim b \text{ for all $a, b \in X \backslash Y$ and $a \succ b$ for all $a \in Y$ and $b \in X \backslash Y$}\}$. 

For any $\succsim_J$, if $\succsim_i \in \mathcal{R} \mid_Y$ for all $i \in J$, then $f(\succsim_J) \in \mathcal{R}\mid_Y$, since Pareto indifference preserves indifference over $X \backslash Y$ and weak Pareto preserves that $a \succ_{f(\succsim_J)} b$ for all $a \in Y, b \in X \backslash Y$. Next note that for each $\succsim_i \in \mathcal{R}\mid_Y$, there is a unique weak order on $Y$ that agrees with $\succsim_i$ on $Y$, and vice versa.

If there exists a social welfare function on $\mathcal{R}$ that satisfies our axioms, then there exists a social welfare function on weak orders of $Y$ that satisfies our axioms (where neutrality concerns permutations of $Y$ with identity on $X \backslash Y$). 

Thus, assume that $X = \{w,x,y,z\}$. Then $\mathcal{R}$ is finite.

\emph{Step 2: semilattice.} By the two remarks above, $f$ satisfies unanimity. Then by Lemma \ref{lemma.chambers} there is a meet-semilattice $(\mathcal{R}, \trianglerighteq)$ with $f(\succsim_J) = \bigwedge_{i \in J} \succsim_i$. Note that $p \trianglerighteq q \iff p \wedge q = q$.

Let's define several specific preferences:
\begin{align*}
        \ell = \pf{w\\x\\y\\z}, \quad m = \pf{w\\x,y\\z}, \quad b = \pf{w\\x\\y,z}, \quad t = \pf{w,x\\y\\z}, \quad s = \pf{w,x\\y,z}.
\end{align*}
One can think of these preferences as collapsing various strict preferences in $\ell$ into indifferences. These are named after (l)inear, (m)iddle, (b)ottom, (t)op, and (s)quare.

\emph{Step 3: four forced meets.}

Weak Pareto, anonymity, and neutrality give the following meets:
\begin{align}
    \ell \wedge \pf{w\\y\\x\\z} &= m, \label{eq.ordf1} \\
    \pf{w\\x\\z\\y} \wedge \pf{x\\w\\y\\z} &= s, \label{eq.ordf2}\\
    m \wedge \pf{x\\w,z\\y} &= s, \label{eq.ordf3}\\
    \ell \wedge \pf{y\\z\\w\\x} &= \pf{w,y\\x,z}. \label{eq.ordf4}
\end{align}
In \eqref{eq.ordf1}, the transposition $(x\,y)$ exchanges the two preferences, so by neutrality $x \sim y$ in the meet, and weak Pareto gives the rest. In \eqref{eq.ordf2} and \eqref{eq.ordf3}, the permutation $(w\,x)(y\,z)$ exchanges the two preferences in each equation, so again by neutrality $w \sim x$ and $y \sim z$ in the meet, and by weak Pareto $w,x$ are strictly preferred to $y,z$. In \eqref{eq.ordf4}, the permutation $(w\,y)(x\,z)$ exchanges the two preferences, so $w \sim y$ and $x \sim z$ in the meet, and by weak Pareto $w$ ($y$) is strictly preferred to $x$ ($z$). 

\emph{Step 4: neither $\pf{w\\x\\z\\y} \trianglerighteq m$ nor $\pf{x\\w\\y\\z} \trianglerighteq m$.} 

If $\pf{w\\x\\z\\y} \trianglerighteq m$, then by the associativity of $\wedge$, \eqref{eq.ordf2} and \eqref{eq.ordf3},
\begin{align*}
        m \wedge \pf{x\\w\\y\\z} = \left(m \wedge \pf{w\\x\\z\\y}\right) \wedge \pf{x\\w\\y\\z} = m \wedge \left(\pf{w\\x\\z\\y} \wedge \pf{x\\w\\y\\z}\right) = m \wedge s = s .
\end{align*}
$m$ and $\pf{x\\w\\y\\z}$ both strictly prefer $y$ to $z$ and $s$ is indifferent, contradicting weak Pareto. Similarly, if $\pf{x\\w\\y\\z} \trianglerighteq m$, then the same equations give $m \wedge \pf{w\\x\\z\\y} = m \wedge s = s$; but $m$ and $\pf{w\\x\\z\\y}$ both strictly prefer $w$ to $x$ and $s$ is indifferent, again contradicting weak Pareto.

\emph{Step 5: $m \wedge \pf{x\\w\\y\\z} = t$ and $m \wedge \pf{w\\x\\z\\y} = b$.} 

Let $p = m \wedge \pf{x\\w\\y\\z}$. By weak Pareto, $p$ ranks $z$ strictly last and $w$ above $y$, and so  
\begin{align*}
    p \in \left\{ \ell, \pf{w\\y\\x\\z}, \pf{x\\w\\y\\z}, m, t \right\}.
\end{align*}

Since $\pf{x\\w\\y\\z} \trianglerighteq p$, Lemma \ref{lemma.comp} rules out $\ell$ and $\pf{w\\y\\x\\z}$. If $p = \pf{x\\w\\y\\z}$, then by \eqref{eq.ordf1} $\ell \trianglerighteq m \trianglerighteq \pf{x\\w\\y\\z}$, contradicting Lemma \ref{lemma.comp}. Step 4 rules out $m$, since it's not the case that $\pf{x\\w\\y\\z} \trianglerighteq m$. Hence $p = t$. Similarly, let $p' = m \wedge \pf{w\\x\\z\\y}$. By weak Pareto, $p'$ ranks $w$ strictly first and $x$ above $z$, so $p' \in \left\{\ell, \pf{w\\y\\x\\z}, \pf{w\\x\\z\\y}, m, b\right\}$, and the same three arguments leave $p' = b$.

\emph{Step 6: reflection and contradiction.} Let $p = \ell \wedge \ell^{-1}$. The permutation $(w\,z)(x\,y)$ exchanges $\ell$ and $\ell^{-1}$, so by neutrality $w \sim_p z$ and $x \sim_p y$. Then, 
\begin{align*}
       p \in \left\{ {\sim^*},\; \pf{x,y\\w,z},\; \pf{w,z\\x,y} \right\}.
\end{align*}
 
\emph{Case 1: $p = {\sim^*}$.} 

By neutrality with the permutation $(w\,y)(x\,z)$,  $\pf{y\\z\\w\\x} \wedge \pf{x\\w\\z\\y} = {\sim^*}$. By $\ell \trianglerighteq {\sim^*}$, associativity of the meet, and \eqref{eq.ordf4},
\begin{align*}
    {\sim^*} = \ell \wedge \left( \pf{y\\z\\w\\x} \wedge \pf{x\\w\\z\\y} \right) = \left( \ell \wedge \pf{y\\z\\w\\x} \right) \wedge \pf{x\\w\\z\\y} = \pf{w,y\\x,z} \wedge \pf{x\\w\\z\\y}.
\end{align*}

Both $\pf{w,y\\x,z}$ and $\pf{x\\w\\z\\y}$ rank $w$ above $z$, contradicting weak Pareto.

\emph{Case 2: $p = \pf{x,y\\w,z}$.} 

By neutrality with transposition $(x\,y)$, $\pf{w\\y\\x\\z} \wedge \pf{z\\x\\y\\w}= p$. Then by \eqref{eq.ordf1} and associativity of the meet,
\begin{align}\label{eq.case2}
    m \wedge \pf{z\\x\\y\\w} = \left( \ell \wedge \pf{w\\y\\x\\z} \right) \wedge \pf{z\\x\\y\\w} = \ell \wedge p = p = \pf{x,y\\w,z}.
\end{align}

By neutrality, the transposition $(w\,x)$ in \eqref{eq.case2} and in the first equality of Step 5 gives $\pf{x\\w,y\\z} \wedge \pf{z\\w\\y\\x} = \pf{w,y\\x,z}$ and $ \pf{x\\w,y\\z} \wedge \ell = t$. Hence, by associativity and \eqref{eq.ordf4},
\begin{align*}
        \pf{w,y\\x,z} = \ell \wedge \pf{w,y\\x,z} = \ell \wedge \left( \pf{x\\w,y\\z} \wedge \pf{z\\w\\y\\x} \right) = \left( \ell \wedge \pf{x\\w,y\\z} \right) \wedge \pf{z\\w\\y\\x} = t \wedge \pf{z\\w\\y\\x}.
\end{align*}
But $t$ and $\pf{z\\w\\y\\x}$ both rank $w$ above $y$, contradicting weak Pareto.

\emph{Case 3: $p = \pf{w,z\\x,y}$.} 

This is the mirror image of Case 2. Exactly as in \eqref{eq.case2}, $m \wedge \pf{z\\x\\y\\w} = \pf{w,z\\x,y}$. The transposition $(y\,z)$ in this equality and in the second equality of Step 5 gives $\pf{w\\x,z\\y} \wedge \pf{y\\x\\z\\w} = \pf{w,y\\x,z}$ and $\pf{w\\x,z\\y} \wedge \ell = b$. Hence, by associativity and \eqref{eq.ordf4},
\begin{align*}
    \pf{w,y\\x,z} = \ell \wedge \pf{w,y\\x,z} = \ell \wedge \left( \pf{w\\x,z\\y} \wedge \pf{y\\x\\z\\w} \right) = \left( \ell \wedge \pf{w\\x,z\\y} \right) \wedge \pf{y\\x\\z\\w} = b \wedge \pf{y\\x\\z\\w}.
\end{align*}

But $b$ and $\pf{y\\x\\z\\w}$ both rank $x$ above $z$, contradicting weak Pareto.

\end{proof}

%% file: draft/proof_rank.tex
\subsection{Proof of Proposition \ref{prop.rank}}

We first prove a lemma that shows that when $|X| = 3$, $f_{trp}$ is a partial priority rule with meet-semilattice $(\mathcal{R}, \trianglerighteq^*)$, for a particular $\trianglerighteq^*$ we define below. Note that this is not true when $|X| \geq 4$. 

For each nonempty $Y \subseteq X$, let $t_Y$ be the weak order whose top indifference class is $Y$ and whose remaining alternatives, if any, form a single indifference class below $Y$. For distinct $a, b \in X$, we write $t_a := t_{\{a\}}$ and $t_{ab} := t_{\{a,b\}}$; note that $t_X = {\sim^*}$.
\begin{lemma}\label{lemma.trp}
    Suppose $|X| = 3$. Define the relation $\trianglerighteq^*$ on $\mathcal{R}$ by
    \begin{align*}
        p \trianglerighteq^* q \iff q = p, \text{ or } q = t_U \text{ for some } U \supseteq I_1(p).
    \end{align*}
    Then $(\mathcal{R}, \trianglerighteq^*)$ is a meet-semilattice. For any nonempty $R \subseteq \mathcal{R}$, let $Y = \bigcup_{p \in R} I_1(p)$. The meet of $R$ is $\ell$ if $R = \{\ell\}$ for some linear order $\ell$, and $t_{Y}$ otherwise. Moreover, for any $J \in \mathcal{I}$ and $\succsim_J \in \mathcal{R}^J$, $f_{trp}(\succsim_J)$ is the meet of $\{\succsim_i : i \in J\}$.
\end{lemma}

\begin{proof}

    \emph{The order.} $\trianglerighteq^*$ is reflexive by definition. It is transitive: if $p \trianglerighteq^* q \trianglerighteq^* r$ with $p \neq q \neq r$, then $q = t_U$ with $U \supseteq I_1(p)$ and $r = t_V$ with $V \supseteq U$, so $V \supseteq I_1(p)$ and $p \trianglerighteq^* r$. It is antisymmetric: if $p \trianglerighteq^* q$ and $q \trianglerighteq^* p$ with $p \neq q$, then $p = t_V$ and $q = t_U$ with $U \supseteq V \supseteq U$, so $p = q$. Thus $\trianglerighteq^*$ is a partial order.

    \emph{Meets.} Let $R \subseteq \mathcal{R}$ be nonempty and finite, and let $Y = \bigcup_{p \in R}I_1(p)$. A non-linear $t_Z$ is a lower bound of $R$ if and only if $Z \supseteq I_1(p)$ for every $p \in R$, that is, if and only if $Z \supseteq Y$.
    
    If $R = \{\ell\}$ for some linear $\ell$, then $\ell$ is a lower bound of $R$ by reflexivity, and every lower bound $q$ of $R$ satisfies $\ell \trianglerighteq^* q$, so the meet is $\ell$. Otherwise, the lower bounds of $R$ are the weak orders $t_Z$ with $Z \supseteq Y$, and the greatest lower bound is $t_Y$. And so the meet of $R$ is $t_{\bigcup_{p \in R}I_1(p)}$. In particular, note that every pair of preferences has a meet, so $(\mathcal{R}, \trianglerighteq^*)$ is a meet-semilattice.

    \emph{The rule.} Let $R$ be the set of preferences in $\succsim_J$, and let $Y = \bigcup_{p \in R}I_1(p)$. By the definition of $f_{trp}$, $I_1(f_{trp}(\succsim_J)) = Y$. Let's consider two cases:
    \begin{itemize}
        \item If $Y$ has at least two elements, at most one alternative is not in $Y$, so $f_{trp}(\succsim_J) = t_{Y}$.
        \item If $Y= \{a\}$ for some $a \in X$, every individual ranks $a$ strictly first. If $|I_2(f_{trp}(\succsim_J))| = 1$, then every individual ranks the remaining alternatives $b$ and $c$ in the same strict order, i.e., all individuals have the same linear order $\ell$, and so $f_{trp}(\succsim_J) = \ell$. Otherwise $I_2(f_{trp}(\succsim_J)) = \{b,c\}$, and $f_{trp}(\succsim_J) = t_a$.
    \end{itemize}
    In each case $f_{trp}(\succsim_J)$ is the meet of $R$.
\end{proof}

\proprank*

\emph{If.} 

\begin{proof}

We prove the result for $f_{trp}$; the argument for $f_{brp}$ is symmetric.

\emph{Representative consistency.} By Lemma \ref{lemma.trp}, $f_{trp}$ is a partial priority rule, so it satisfies representative consistency. We also provide this explicit demonstration for illustration:

Consider any $J \in \mathcal{I}$, $\succsim_J \in \mathcal{R}^J$, and partition $\{K_1, \dots, K_m\}$ of $J$, and let $\succsim_0' := f_{trp}(f_{trp}(\succsim_{K_1})^{K_1}, \dots, f_{trp}(\succsim_{K_m})^{K_m})$. We show that $f_{trp}(\succsim_J)$ and $\succsim_0'$ have the same indifference classes.

\emph{First indifference class}. Consider any $a \in I_1(f_{trp}(\succsim_J))$. By definition, there exists $i \in J$ for whom $a \in I_1(\succsim_i)$. Write $K_\ell$ for the block of the partition that contains $i$. By definition of $f_{trp}$, $a \in I_1(f_{trp}(\succsim_{K_\ell}))$. Again by definition, $a \in I_1(\succsim_0')$.

\emph{Second indifference class}. Consider any $a \in I_2(f_{trp}(\succsim_J))$. By definition of $f_{trp}$, there does not exist an individual $i$ for whom $a \in I_1(\succsim_i)$. Again by definition, for each $K_\ell$, $a \notin I_1(f_{trp}(\succsim_{K_\ell}))$, and so $a \notin I_1(\succsim_0')$. Assume towards a contradiction that $a \in I_3(\succsim_0')$. Since there are three alternatives in this case, $I_1(\succsim_0')$ and $I_2(\succsim_0')$ are singletons. Label them $\{b\}$ and $\{c\}$ respectively. Then by definition of $f_{trp}$, for each $K_\ell$, $f_{trp}(\succsim_{K_\ell}) = \pf{b\\c\\a}$, and so for each $i$, $\succsim_i = \pf{b\\c\\a}$. This contradicts that there exists an individual $j$ for whom $a \in I_1(\succsim_j\mid_{X - I_1(f_{trp}(\succsim_J))})$. Thus, $a \in I_2(\succsim_0')$.

\emph{Third indifference class}. Consider $a \in I_3(f_{trp}(\succsim_J))$. Since there are only three alternatives, $I_1(f_{trp}(\succsim_J))$ and $I_2(f_{trp}(\succsim_J))$ are singletons. Label them $\{b\}$ and $\{c\}$ respectively. By definition of $f_{trp}$, $\succsim_i = \pf{b\\c\\a}$ for all $i \in J$. Since $f_{trp}$ of a unanimous linear profile is that linear order, $f_{trp}(\succsim_{K_\ell}) = \pf{b\\c\\a}$ for each $K_\ell$, and $\succsim_0' = \pf{b\\c\\a}$.

Since $I_k(f_{trp}(\succsim_J)) \subseteq I_k(\succsim_0')$ for $k = 1,2,3$ and there are at most three indifference classes, $f_{trp}(\succsim_J) = \succsim_0'$.

\emph{Weak Pareto.} Consider any $J \in \mathcal{I}$ and $\succsim_J \in \mathcal{R}^J$. Suppose that $x \succ_i y$ for all $i \in J$. Let $k$ be such that $I_k(f_{trp}(\succsim_J))$ contains $x$. Assume towards a contradiction that $y \in I_\ell(f_{trp}(\succsim_J))$ for $\ell \leq k$. By the definition of $f_{trp}$, there exists some $j$ for whom $y \succsim_j x$. This contradicts that $x \succ_i y$ for all $i \in J$.

Since $f_{trp}$ does not mention specific individuals or alternatives, it satisfies anonymity and neutrality.

\end{proof}

\emph{Only if.}

\begin{proof}

By the two remarks preceding Lemma \ref{lemma.chambers}, $f$ satisfies unanimity, and so by Lemma \ref{lemma.chambers}, $f$ is a partial priority rule with some meet-semilattice $(\mathcal{R}, \trianglerighteq)$. Write $p \wedge q := f(p,q)$, so that $p \trianglerighteq q \iff p \wedge q = q$. We will show that ${\trianglerighteq} = {\trianglerighteq^*}$.

First note the following fact, which we refer to as ($\bot$): ${\sim^*}$ is the least element, i.e. $p \trianglerighteq {\sim^*}$ for all $p \in \mathcal{R}$. Indeed, the meet of all of $\mathcal{R}$ must be the same under all permutations by neutrality, and ${\sim^*}$ is the only weak order that remains the same under all permutations.

    The case $|X| = 1$ is trivial. 
    
    If $X = \{x,y\}$, there are only three preferences: $\pf{x\\y}, \pf{y\\x}, {\sim^*}$. $\pf{x\\y} \wedge \pf{y\\x} = {\sim^*}$ by neutrality. This determines $\trianglerighteq$, and so $f$ coincides with $f_{trp}$ and $f_{brp}$, which coincide with each other here.
    
    Suppose from now on that $X = \{x,y,z\}$. We use the notation introduced before Lemma \ref{lemma.trp}: for distinct $a,b \in X$, $t_a$ is the weak order with $a$ strictly on top and the other two alternatives indifferent, and $t_{ab}$ is the weak order with $a$ and $b$ indifferent and above the third alternative. Let $\ell = \pf{x\\y\\z}$.

    \emph{Step 1: forced meets.} The transpositions $(x\,y)$ and $(y\,z)$ give, by neutrality and weak Pareto,
    \begin{align*}
        \ell \wedge \pf{y\\x\\z} & = t_{xy},\\
        \ell \wedge \pf{x\\z\\y} & = t_x.
    \end{align*}
    Then by neutrality, for any linear order $p$,
    \begin{align}
        p & \trianglerighteq t_{ab} \quad \text{when $I_3(p) = \{c\}$ with $c \neq a, b$}, \label{eq.rp1}\\
        p & \trianglerighteq t_a \quad \text{when $I_1(p) = \{a\}$}. \label{eq.rp2}
    \end{align}

    \emph{Step 2: the two branches.} 
    
    Let $p = t_x \wedge t_{xz}$. Both preferences rank $x$ strictly above $y$, so by weak Pareto,
    \begin{align*}
        p \in \left\{\pf{x\\y\\z}, \pf{x\\z\\y},  \pf{z\\x\\y}, t_x, t_{xz} \right\}.
    \end{align*}
    If $p = \pf{x\\y\\z}$ or $p = \pf{x\\z\\y}$, then $t_x \trianglerighteq p$, while $p \trianglerighteq t_x$ by Step 1, contradicting antisymmetry. If $p = \pf{z\\x\\y}$, then $t_{xz} \trianglerighteq p$, while $p \trianglerighteq t_{xz}$ by Step 1, again contradicting antisymmetry. So $t_x$ and $t_{xz}$ are $\trianglerighteq$-comparable. Whether $f$ is $f_{trp}$ or $f_{brp}$ depends on this comparison. 

    We argue that it suffices to treat the case $t_x \trianglerighteq t_{xz}$. The rule $\succsim_J \mapsto f(\succsim_J^{-1})^{-1}$ satisfies the four axioms and is the partial priority rule of the reversed semilattice. Reversal maps $t_x$ to $t_{yz}$ and $t_{xz}$ to $t_y$, so it exchanges the case $t_{xz} \trianglerighteq t_x$ with the case $t_y \trianglerighteq t_{yz}$, which by neutrality is the case $t_x \trianglerighteq t_{xz}$. Since the reversal of $f_{trp}$ is $f_{brp}$, it suffices to show that if $t_x \trianglerighteq t_{xz}$, then $f = f_{trp}$. 
    
    So assume $t_x \trianglerighteq t_{xz}$. By neutrality,
    \begin{align}
        t_a \trianglerighteq t_{ab} \quad \text{for distinct } a, b \in X. \label{eq.rp3}
    \end{align}
    
    Steps 1, 2, and ($\bot$) show that $p \trianglerighteq^* q$ implies $p \trianglerighteq q$. 

    Let's show the other direction for equivalence. We must characterize the rest of the meets for $\trianglerighteq$.

    \emph{Step 3: three meets.}

    Let $p := t_x \wedge t_y $. By neutrality, $p \in \{ t_z, {\sim^*}, t_{xy}\}$. If $p = t_z$, then $t_x \trianglerighteq t_z$, contradicting Lemma \ref{lemma.comp}. By \eqref{eq.rp3}, $t_x \wedge t_y \trianglerighteq t_{xy}$. If $p = {\sim^*}$, then ${\sim^*} \trianglerighteq t_{xy}$, contradicting ($\bot$). Then $p = t_{xy}$, and neutrality gives
\begin{align} \label{eq.rp4}
    t_a \wedge t_b = t_{ab} \quad \text{for distinct } a, b \in X.
\end{align}
    By Step 1, associativity, \eqref{eq.rp2}, and \eqref{eq.rp3},
    \begin{align*}
        t_x \wedge \pf{y\\x\\z} = \left( t_x \wedge \ell \right) \wedge \pf{y\\x\\z} = t_x \wedge \left( \ell \wedge \pf{y\\x\\z} \right) = t_x \wedge t_{xy} = t_{xy}.
    \end{align*}
    Neutrality gives 
    \begin{align} \label{eq.rp5}
        t_a \wedge \pf{b\\a\\c} = t_{ab} \quad \text{for distinct } a, b, c \in X.
    \end{align}

    Let $p:= \ell \wedge \pf{z\\x\\y}$. Both preferences rank $x$ strictly above $y$, so by weak Pareto,
    \begin{align*}
        p \in \left\{\pf{z\\x\\y}, \pf{x\\z\\y}, \ell, t_x, t_{xz}\right\}.
    \end{align*}
    The three linear orders are ruled out by Lemma \ref{lemma.comp}. If $p = t_x$, then $\pf{z\\x\\y} \trianglerighteq t_x$. But \eqref{eq.rp5} gives $t_x \wedge \pf{z\\x\\y} = t_{xz} \neq t_x$. Hence $p = t_{xz}$, and neutrality gives
    \begin{align} \label{eq.rp6}
        \pf{a\\b\\c} \wedge \pf{c\\a\\b} = t_{ac} \quad \text{for distinct }a, b, c \in X.
    \end{align}

    \emph{Step 4: $p \trianglerighteq q$ implies $p \trianglerighteq^* q$.}
    
    Let $p \trianglerighteq q$ with $p \neq q$. By ($\bot$), $p \neq {\sim^*}$. Also, $q$ is not linear. Otherwise $p$ is not linear by Lemma \ref{lemma.comp}, and so by Step 1 some linear order $r \neq q$ satisfies $r \trianglerighteq p \trianglerighteq q$, contradicting Lemma \ref{lemma.comp}.

    We examine the remaining cases:
    \begin{enumerate}
        \item[(i)] If $p = t_{ab}$, then $q = {\sim^*}$. Suppose $p = t_{xy}$ and $q \neq {\sim^*}$. Since $q$ is not linear, by Lemma \ref{lemma.comp}, $q = t_a$ for some $a \in X$. But then $t_x \trianglerighteq t_{xy} \trianglerighteq t_a$ by \eqref{eq.rp3}, contradicting Lemma \ref{lemma.comp} if $a \neq x$ and antisymmetry if $a = x$. Neutrality gives the statement for $p = t_{ab}$.
        \item[(ii)] If $p = t_c$ and $q = t_{ab}$, then $c \in \{a,b\}$. Otherwise $t_a \trianglerighteq t_{ab}$ and $t_c \trianglerighteq t_{ab}$ give $t_a \wedge t_c \trianglerighteq t_{ab}$. But $t_a \wedge t_c = t_{ac}$ by \eqref{eq.rp4}, so $t_{ac} \trianglerighteq t_{ab}$, contradicting Lemma \ref{lemma.comp}.
        \item[(iii)] If $p$ is linear with $I_1(p) = \{a\}$ and $q = t_b$, then $b = a$. Suppose $b \neq a$. If $b$ is the middle alternative of $p$, then $p \wedge t_b = t_{ab} \neq t_b$ by \eqref{eq.rp5}. If $b$ is the bottom alternative of $p$, i.e., $p = \pf{a\\c\\b}$, let $p' = \pf{b\\a\\c}$. Then $p' \trianglerighteq t_b$ by Step 1, so $p \wedge p' \trianglerighteq t_b$, and $p \wedge p' = t_{ab}$ by \eqref{eq.rp6}. Hence $t_{ab} \trianglerighteq t_b$, which together with \eqref{eq.rp3} contradicts antisymmetry.
        \item[(iv)] If $p$ is linear with $I_1(p) = \{a\}$ and $q = t_{bc}$, then $a \in \{b, c\}$. Otherwise $b$ and $c$ are the two lower alternatives of $p$ --- suppose $b$ is the middle one. Since $p \trianglerighteq t_{bc}$ and $t_b \trianglerighteq t_{bc}$, we have $p \wedge t_b \trianglerighteq t_{bc}$, and $p \wedge t_b = t_{ab}$ by \eqref{eq.rp5}. Hence $t_{ab} \trianglerighteq t_{bc}$, contradicting Lemma \ref{lemma.comp}.
    \end{enumerate}
    In each case $q$ is non-linear with $I_1(q) \supseteq I_1(p)$, so $p \trianglerighteq^* q$.

    By Steps 2 and 4, ${\trianglerighteq} = {\trianglerighteq^*}$. Since $f$ and $f_{trp}$ both assign to each profile the meet of its preferences in this semilattice (Lemma \ref{lemma.trp}), $f = f_{trp}$.

\end{proof}

%% file: draft/proof_uncertain.tex
\section{Proofs of Results in Section \ref{sec.eu}}

\subsection{Proof of Theorem \ref{thm.eu}}

We make use of a variant of \citeauthor{harsanyi1955cardinal}'s Aggregation Theorem.

\begin{lemma} \label{lemma.harsanyi}
    \citep{weymark1993harsanyi, de1995note}\footnote{\cite{weymark1993harsanyi} gives this result when $X$ is finite. \cite{de1995note} gives this result for collections of functions $f_1, \dots, f_n$ on some set $A$ for which the image of $(f_1, \dots, f_n)$ on $A$ is convex. The case where $A$ is $\Delta(X)$ and $f_1, \dots, f_n$ are vNM expected utility functions satisfies this convexity requirement.} Suppose $\phi$ satisfies weak Pareto. For any $J \in \mathcal{I}$ and $\succsim_J \in \mathcal{V}^J$, there exist a nonnegative weight $\eta$, semipositive weights $\{\lambda_i\}_{i \in J}$ ($\lambda_i$ nonnegative and not all zero), and $\mu \in \mathbb{R}$ such that
    \begin{align*}
        \eta u_0(x) = \sum_{i \in J} \lambda_i u_i(x) + \mu \text{ for all } x \in X,
    \end{align*}
    where $u_0$ is a utility index in the representation of $\phi(\succsim_J)$ and $u_i$ is that of $\succsim_i$.
\end{lemma}

Note that if $\succsim_j = \succsim_k$ for some distinct $j, k\in J$, then $u_k = \alpha u_j + \beta$ for some $\alpha > 0$ and $\beta \in \mathbb{R}$. Then $\eta u_0(x) = \lambda_j u_j(x) + \lambda_k u_k(x) + \sum_{i \in J \backslash \{j,k\}} \lambda_i u_i(x) + \mu = (\lambda_j + \lambda_k\alpha)u_j(x) + \sum_{i \in J \backslash \{j,k\}} \lambda_i u_i(x) + (\mu + \lambda_k\beta)$. If $\{\lambda_i\}_{i \in J}$ are semipositive, then $\{\lambda_i\}_{i \in J \backslash \{j,k\}} \cup \{\lambda_j + \lambda_k \alpha\}$ are also semipositive, and so the equation in Lemma \ref{lemma.harsanyi} can be stated with one weight for each distinct preference in the profile.

\begin{remark}
    If $\phi$ satisfies weak Pareto and neutrality, then $\phi$ satisfies unanimity.
\end{remark}
Consider any $J \in \mathcal{I}$ and $p \in \mathcal{V}$, and let $q = \phi(p^J)$. By Lemma \ref{lemma.harsanyi} and since every individual has the same preferences, $\eta u_q = \lambda u_p + \mu$ for some $\eta \geq 0$, $\lambda > 0$, and $\mu \in \mathbb{R}$. If $p \neq {\sim^*}$, then $u_q$ is a positive affine transformation of $u_p$, and so $q = p$. If $p = {\sim^*}$, $\phi({\sim^{*J}}) = {\sim^*}$ by neutrality, since all permutations preserve the profile, and ${\sim^*}$ is the only preference relation that is invariant to all permutations.

\thmeu*

\begin{proof}

Pick any three distinct alternatives $x, y, z \in X$. Let $Y = X \backslash \{x,y,z\}$.

Write $\vpf{x\\y,Y\\z}\in \mathcal{V}$ for the vNM preference that is represented by the utility index
\begin{align*}
    u(w) = \begin{cases}
        1 & \text{if $w=x$} \\
        0 & \text{if $w=y$ or $w \in Y$} \\
        -1 & \text{if $w=z$,}
    \end{cases}
\end{align*} 
and analogously for other labelings of $x, y, z$. For $\succsim \in \mathcal{V}$, write $\succsim^{-1}$ for the reversed preference, i.e., $P \succsim^{-1} Q \iff Q \succsim P$ for any $P, Q \in \Delta(X)$; if $u$ represents $\succsim$, then $-u$ represents $\succsim^{-1}$. For instance, $\vpf{x\\y,Y\\z}^{-1} = \vpf{z\\y,Y\\x}$. 

Since $\phi$ satisfies anonymity, we write each profile as a list of preferences without specifying the individuals, and in partitions we refer to individuals by their position in the list.

\emph{Step 1: two preferences.}

Let $p = \vpf{x\\y,Y\\z}, q = \vpf{z\\x,Y\\y}$, and $s = \vpf{y\\z,Y\\x}$, and let
\begin{align*}
    c = \phi\left(p, q, s\right).
\end{align*}

The permutation $(x\,y\,z)$ permutes the three preferences, and so by neutrality and anonymity $x \sim_c y \sim_c z$. The notation $c$ is for (c)ondorcet.
    
Next, let
\begin{align*}
    r = \phi\left(p, p^{-1}\right).
\end{align*}

The transposition $(x\,z)$ exchanges the two preferences, and so by neutrality and anonymity, $x \sim_r z$. The notation $r$ is for (r)eflection.

\emph{Step 2: representative consistency.}

Applying representative consistency with partition $\{\{1\}, \{2,3\}\}$ and unanimity to the profile above gives:

\begin{align*}
    c = \phi\left(\phi\left(p\right), \phi\left(q,s\right), \phi\left(q,s\right)\right) = \phi\left(p, t, t\right),
\end{align*}
where $t:=\phi\left(q,s\right)$ for (t)urn.

\emph{Step 3: $t = p^{-1}$.}

Applying Lemma \ref{lemma.harsanyi} to $t=\phi\left(q,s\right)$, there exist $\eta_t \geq 0$, semipositive $\lambda_q$ and $\lambda_s$, and $\mu \in \mathbb{R}$ such that 
\begin{align*}
    \eta_t u_t(a) = \lambda_q u_q(a) + \lambda_s u_s(a)  + \mu \text{ for all $a \in X$}.
\end{align*}

Since $z \succ_q x$ and $z \succ_s x$, $\lambda_q u_q + \lambda_s u_s$ cannot be constant on $X$. And so $\eta_t > 0$. Since $u_t$ is unique up to positive affine transformations, we can drop $\eta_t$ in the expression (rescaling $\lambda_q$, $\lambda_s$, and $\mu$):
\begin{align*}
    u_t(a) = \lambda_q u_q(a) + \lambda_s u_s(a)  + \mu \text{ for all $a \in X$}.
\end{align*}

Applying Lemma \ref{lemma.harsanyi} to $c = \phi(p, t,t)$, there exist $\eta_c \geq 0$, semipositive $\lambda_p$ and $\lambda_t$, and $\mu' \in \mathbb{R}$ such that 
\begin{align*}
    \eta_c u_c(a) = \lambda_p u_p(a) + \lambda_{t} u_{t}(a)  + \mu' \text{ for all $a \in X$}.
\end{align*}

Step 1 gives $u_c(x) = u_c(y) = u_c(z)$. Then $\lambda_p u_p(a) + \lambda_{t}\left(\lambda_q u_q(a) + \lambda_s u_s(a) + \mu \right) + \mu'$ gives the same value for $a = x, y, z$. By the definition of the utilities, $\lambda_p - \lambda_t\lambda_s = \lambda_t(\lambda_s - \lambda_q) = \lambda_t\lambda_q - \lambda_p$. Note that these three expressions sum to $0$, so each expression equals $0$. Thus $\lambda_p = \lambda_t \lambda_q = \lambda_t \lambda_s$.

We claim that $\lambda_p > 0$. Otherwise if $\lambda_p = 0$, then $\lambda_t > 0$, which then requires that $\lambda_q = \lambda_s = 0$, a contradiction. So $\lambda_p > 0$, which then implies that $\lambda_t > 0$ and $\lambda_q = \lambda_s$.

Then, $u_t = \lambda_q(u_q + u_s ) + \mu = \lambda_q u_{p^{-1}} + \mu$. Utility indices are unique up to positive affine transformation, so $t = p^{-1}$.

\emph{Step 4: $x \sim_r y \sim_r z$.}

By Steps 2 and 3, $c = \phi(p, p^{-1}, p^{-1})$. Applying representative consistency with partition $\{\{1,2\}, \{3\}\}$ and unanimity yields $c = \phi(r, r, p^{-1})$. 

Applying Lemma \ref{lemma.harsanyi} to $c = \phi(r, r, p^{-1})$, there exist $\eta_c' \geq 0$, semipositive $\lambda_r$ and $\lambda_{p^{-1}}$, and $\mu'' \in \mathbb{R}$ such that
\begin{align*}
    \eta_c' u_c(a) = \lambda_r u_r(a) + \lambda_{p^{-1}}u_{p^{-1}}(a) + \mu'' \text{ for all $a \in X$}.
\end{align*}

Consider the difference between values at alternatives $z$ and $x$. On the left side, since $x \sim_c z$, the difference is $0$. On the right, $x \sim_r z$ and so $u_r$ gives $0$, while $u_{p^{-1}}$ gives a difference of $2$. Thus, $\lambda_{p^{-1}} = 0$, which implies that $\lambda_r > 0$. Then $\eta_c' u_c = \lambda_r u_r + \mu''$. If $\eta_c' > 0$, then $c = r$. If $\eta_c' = 0$, then $r = {\sim^*}$. Either way, $x \sim_r y \sim_r z$.

\emph{Step 5: Contradiction.}

The permutation $\rho :=(x\,z\,y)$ maps $p \mapsto q$ and $p^{-1} \mapsto q^{-1}$. Then $\phi(q, q^{-1}) = \rho(r) = r$ since $x \sim_r y \sim_r z$.

By representative consistency with partition $\{\{1,2\}, \{3\}\}$ and unanimity,
\begin{align*}
    \phi(q, q^{-1}, p^{-1}) = \phi(r, r, p^{-1}) = c.
\end{align*}

The transposition $\rho' := (y\,z)$ maps $q \mapsto q^{-1}$, $s \mapsto p^{-1}$, and $p^{-1} \mapsto s$, and so since $\phi(q, s) = t = p^{-1}$, we have $\phi(q^{-1}, p^{-1}) = \rho'(t) = \rho'(p^{-1}) = s$. 

By representative consistency with partition $\{\{1\}, \{2,3\}\}$ and unanimity,
\begin{align*}
    \phi(q, q^{-1}, p^{-1}) = \phi(q, s, s).
\end{align*}

These two equations give $\phi(q, s, s) = c$. But $z \succ_q x$, $z \succ_s x$, while $z \sim_c x$, contradicting weak Pareto.

\end{proof}

\subsection{A Consequence of Strong Pareto}
 
The following remark is used in the proofs of Propositions \ref{prop.eustrong} and \ref{prop.xpriority}. It holds for any $X$.
 
\begin{remark}
    Suppose $\phi$ satisfies strong Pareto. For any $J \in \mathcal{I}$, $\succsim \in \mathcal{V}$, and $\succsim_J \in \mathcal{V}^J$ such that $\succsim_i \in \{\succsim, {\sim^*}\}$ for all $i \in J$ and $\succsim_i = \succsim$ for some $i \in J$,
    \begin{equation} \label{eq.unan}
        \phi(\succsim_J) = \succsim.
    \end{equation}
\end{remark}
To see this, for lotteries $P, Q \in \Delta(X)$, if $P \succ Q$, every $i \in J$ weakly prefers $P$ to $Q$ and some $i$ strictly does, so $P \succ_{\phi(\succsim_J)} Q$. If $P \sim Q$, every $i \in J$ is indifferent, so $P \sim_{\phi(\succsim_J)} Q$. Since $\succsim$ is complete, $\phi(\succsim_J) = \succsim$.
 
\subsection{Proof of Proposition \ref{prop.eustrong}}
 
\propeustrong*
 
\begin{proof}
 
Fix three distinct individuals and label them $1, 2$, and $3$. For any $p, q, r \in \mathcal{V}$, write $\phi(p, q, r)$ for the social preference at the profile where individuals $1$, $2$, and $3$ have preferences $p$, $q$, and $r$ respectively, and write $\phi(p, q)$ for the social preference at the profile where individuals $1$ and $2$ have preferences $p$ and $q$, respectively.
 
Since $|X| \geq 3$, fix distinct $x, y, z \in X$. Let $p, q \in \mathcal{V}$ be preferences represented by the indicator utility indices $u_p := \mathbf{1}_{\{x\}}$ and $u_q := \mathbf{1}_{\{y\}}$, and let $r := \phi(p, q)$.
 
\emph{Step 1: $x \succ_r z$ and $y \succ_r z$.} Individual 1 strictly prefers $x$ to $z$ and individual 2 is indifferent between them, so $x \succ_r z$ by strong Pareto. Similarly, individual 2 strictly prefers $y$ to $z$ and individual 1 is indifferent, so $y \succ_r z$.
 
\emph{Step 2: $\phi(r, r, p) = r$.} On the one hand,
\begin{align*}
    \phi(p, q, {\sim^*}) & = \phi(r, r, {\sim^*}) && \text{by representative consistency with $\{\{1,2\},\{3\}\}$ and \eqref{eq.unan}} \\
    & = r && \text{by \eqref{eq.unan}}.
\end{align*}
On the other hand,
\begin{align*}
    \phi(p, q, {\sim^*}) & = \phi(p, q, p) && \text{by representative consistency with $\{\{1,3\},\{2\}\}$ and \eqref{eq.unan}} \\
    & = \phi(r, r, p) && \text{by representative consistency with $\{\{1,2\},\{3\}\}$ and \eqref{eq.unan}}.
\end{align*}
Hence $\phi(r, r, p) = r$.
 
Note that in both computations, the fact that $r = \phi(p,q)$ is used only for the subpopulation $\{1,2\}$, so anonymity is not used.
 
\emph{Step 3: Contradiction.} Let $u_r$ be a utility index for $r$, normalized so that $u_r(z) = 0$. By Step 1, $\alpha := u_r(x) > 0$ and $\beta := u_r(y) > 0$. Define lotteries:
\begin{align*}
    P := \tfrac{\beta}{\alpha + \beta}\, x + \tfrac{\alpha}{\alpha + \beta}\, z
    \quad \text{and} \quad
    Q := \tfrac{\alpha}{\alpha + \beta}\, y + \tfrac{\beta}{\alpha + \beta}\, z.
\end{align*}
Then $u_r(P) = \tfrac{\alpha\beta}{\alpha+\beta} = u_r(Q)$, so $P \sim_r Q$, whereas $u_p(P) = \tfrac{\beta}{\alpha+\beta} > 0 = u_p(Q)$, so $P \succ_p Q$. At the profile $(r, r, p)$, individuals 1 and 2 are indifferent between $P$ and $Q$ and individual 3 strictly prefers $P$, so strong Pareto requires that $P \succ_{\phi(r,r,p)} Q$. This contradicts $\phi(r, r, p) = r$.
\end{proof}
 
\subsection{Proof of Proposition \ref{prop.xpriority}}
 
\propxpriority*
 
\begin{proof}
Write $p_x$ and $p_y$ for the preferences with $x \succ y$ and $y \succ x$, respectively, so that $\mathcal{V} = \{p_x, p_y, {\sim^*}\}$.
 
\emph{If.} We consider the $x$-priority rule; the $y$-priority rule is symmetric. The $x$-priority rule is the total priority rule of Section \ref{sec.ind} for the order $p_x \succ^{\mathcal{V}} p_y \succ^{\mathcal{V}} {\sim^*}$, so it satisfies representative consistency. For strong Pareto, fix $J \in \mathcal{I}$, $\succsim_J \in \mathcal{V}^J$, and $P, Q \in \Delta(X)$ such that $P \succsim_i Q$ for all $i \in J$. If some individual has the preference $p_x$, then $P(x) \geq Q(x)$ and $\phi_x(\succsim_J) = p_x$, so $P \succsim_{\phi_x(\succsim_J)} Q$. If furthermore some individual strictly prefers $P$ to $Q$, then $P(x) > Q(x)$ and $P \succ_{\phi_x(\succsim_J)} Q$. If no individual has $p_x$ but some individual has $p_y$, the same argument applies with $x$ and $y$ exchanged. If every individual has ${\sim^*}$, then $\phi_x(\succsim_J) = {\sim^*}$.
 
\emph{Only if.} Let $\phi$ satisfy representative consistency and strong Pareto. For distinct $i, j \in I$ and $p, q \in \mathcal{V}$, write $\phi_{ij}(p,q)$ for the social preference at the profile where $\succsim_i = p$ and $\succsim_j = q$. Note that $\phi_{ij}(p,q) = \phi_{ji}(q,p)$. Write $\phi_{ijk}(p,q,r)$ similarly for distinct $i,j,k \in I$. Each equation below follows from representative consistency with the indicated partition and \eqref{eq.unan}.
 
\emph{Step 1: $\phi_{ij} = \phi_{kl}$ for any distinct $i, j \in I$ and any distinct $k, l \in I$.} 

For distinct $i,j,k \in I$ and any $p, q \in \mathcal{V}$,
\begin{align*}
    \phi_{ij}(p,q) & = \phi_{ijk}(p,q,{\sim^*}) && \text{partition } \{\{i,j\},\{k\}\} \\
    & = \phi_{ijk}(p,q,p) && \text{partition } \{\{i,k\},\{j\}\} \\
    & = \phi_{ijk}({\sim^*},q,p) && \text{partition } \{\{i,k\},\{j\}\} \\
    & = \phi_{kj}(p,q) && \text{partition } \{\{j,k\},\{i\}\}.
\end{align*}
That is, the individual with preference $p$ can be replaced by any individual other than $j$. Since $\phi_{ij}(p,q) = \phi_{ji}(q,p)$, the individual with preference $q$ can similarly be replaced: $\phi_{ij}(p,q) = \phi_{ji}(q,p) = \phi_{ki}(q,p) = \phi_{ik}(p,q)$. Fix distinct $k, l \in I$. If $k \neq j$, replace $i$ by $k$ and then $j$ by $l$. If $k = j$, pick $m \notin \{i,j\}$ and replace individuals sequentially to move from $(i,j)$ to $(i,m)$, then to $(j,m)$, then to $(j,l)$. Then $\phi$ is anonymous for profiles of size two, i.e. for any $p, q \in \mathcal{V}$, $\phi_{ij}(p,q)$ does not depend on the distinct $i, j \in I$.
 
\emph{Step 2: $\phi(p_x, p_y) \neq {\sim^*}$.} 

Suppose $ \phi(p_x, p_y) = {\sim^*}$, and fix distinct $i, j, k \in I$. Then
\begin{align*}
    {\sim^*} = \phi_{ijk}({\sim^*},{\sim^*},{\sim^*}) & = \phi_{ijk}(p_x,p_y,{\sim^*}) && \text{partition } \{\{i,j\},\{k\}\} \\
    & = \phi_{ijk}(p_x,p_y,p_x) && \text{partition } \{\{i,k\},\{j\}\} \\
    & = \phi_{ijk}({\sim^*},{\sim^*},p_x) = p_x && \text{partition } \{\{i,j\},\{k\}\},
\end{align*}
a contradiction. So $\phi(p_x, p_y) \in \{p_x, p_y\}$. These two possibilities correspond with the $x$-priority and $y$-priority rules. Suppose $\phi(p_x, p_y) = p_x$. If instead $\phi(p_x, p_y) = p_y$, the same argument with $x$ and $y$ exchanged shows that $\phi$ is the $y$-priority rule.
 
\emph{Step 3: $\phi = \phi_x$.} Fix $J \in \mathcal{I}$ and $\succsim_J \in \mathcal{V}^J$. If no individual in $J$ has preference $p_x$, then $\phi(\succsim_J) = \phi_x(\succsim_J)$ by \eqref{eq.unan}. Otherwise, fix $j \in J$ with $\succsim_j = p_x$. We show that $\phi(\succsim_J) = p_x$ by induction on the number $n$ of individuals in $J$ with $p_y$. If $n = 0$, this follows from \eqref{eq.unan}. If $n \geq 1$, pick $k \in J$ with $\succsim_k = p_y$. Representative consistency with the partition of $J$ into $\{j,k\}$ and singletons, together with $\phi_{jk}(p_x,p_y) = p_x$ and \eqref{eq.unan}, gives $\phi(\succsim_J) = \phi(\succsim_J')$, where $\succsim_J'$ agrees with $\succsim_J$ except that $\succsim_k' = p_x$. Note that $\succsim_J'$ has $n-1$ individuals with $p_y$. Then, by induction, $\phi(\succsim_J) = p_x$. 
\end{proof}